\documentclass{article}

\usepackage{vaucanson-g}
\usepackage{pst-tree}

\usepackage{pst-node}

\newtheorem{theorem}{Theorem}
\newtheorem{conjecture}{Conjecture}
\newtheorem{corollary}{Corollary}
\newtheorem{proposition}{Proposition}
\newtheorem{lemma}{Lemma}

\newenvironment{proof}{\begin{trivlist}\item[]{\bf Proof.}}{\hfill $\rule{2mm}{2mm}$ \end{trivlist}}

\begin{document}
\title{Fewest repetitions in infinite binary words}

\author{Golnaz Badkobeh\\
King's College London, London, UK
\and
Maxime Crochemore \\
King's College London, London, UK\\
and Universit\'e Paris-Est, France}

\maketitle

\begin{abstract}

A square is the concatenation of a nonempty word with itself.
A word has period $p$ if its letters at distance $p$ match.
The exponent of a nonempty word is the quotient of its length
 over its smallest period.

In this article we give a proof of the fact that there exists an
 infinite binary word which contains finitely many squares and
 simultaneously avoids words of exponent larger than $7/3$.

Our infinite word contains 12 squares, which is the smallest possible
 number of squares to get the property, and 2 factors of exponent $7/3$.
These are the only factors of exponent larger than $2$.

The value $7/3$ introduces what we call the finite-repetition threshold of
 the binary alphabet.
We conjecture it is $7/4$ for the ternary alphabet, like its repetitive threshold.


\end{abstract}


\section{Introduction}

Repetitions in words is a basic question in Theoretical Informatics, certainly
 because it is related to many applications although it has first been studied
 by Thue at the beginning of the twentieth century \cite{Thue06} with a pure
 theoretical objective.
Related results apply to the design of efficient string pattern matching algorithm,
 to text compression methods and entropy analysis, as well as to the study of
 repetitions in biological molecular sequences among others.

The knowledge of the strongest constraints an infinite word can tolerate
 help for the design and analysis of efficient algorithms.
The optimal bound on the maximal exponent of factors of the word has been
 studied by Thue and many other authors after him.
One of the first discoveries was that an infinite binary word can avoid factors
 with an exponent larger than $2$, called $2^+$-powers.
This has been extended by Dejean \cite{Dej72} to the ternary alphabet
 and her famous conjecture on the repetitive threshold for larger alphabets
 has eventually been proved recently after a series of partial results by
 different authors (see \cite{Rao11,CR11} and references therein).

Another constraint is considered by Fraenkel and Simpson \cite{FS95}:
 their parameter to the complexity of binary infinite words is the number of
 squares occurring in them without any restriction on the number of occurrences.
It is fairly straightforward to check that no infinite binary word can contain
 less than three squares and they proved that some of them contain exactly three.
Two of these squares appear in the cubes $\texttt{000}$ and $\texttt{111}$ so
 that the maximum exponent is $3$ in their word.
In this article we produce an infinite word with few distinct squares and
 a smaller maximal exponent.

Fraenkel and Simpson's proof uses a pair of morphisms, one to get an infinite word by iteration,
 the other to produce the final translation on the binary alphabet.
Their result has been proved with different pairs of morphisms
 by Rampersad et al. \cite{RampersadSW05} (the first morphism is uniform),
 by Harju and Nowotka \cite{HN06} (the second morphism accepts any infinite
 square-free word),
 and by Badkobeh and Crochemore \cite{BC10-3sq} (the simplest morphisms).

In this article we show that we can combine the two types of constraints for the
 binary alphabet:
 producing an infinite word whose maximal exponent of its factor is the smallest
 possible while containing the smallest number of squares.
The maximal exponent is $7/3$ and the number of squares is $12$ to which can be
 added two words of exponent $7/3$.

It is known from Karhum{\"a}ki and Shallit \cite{KarhumakiS04} that if an
 infinite binary word avoids $7/3$-powers it contains an infinite number of squares.
Proving that it contains more than $12$ squares is indeed a matter of simple
 computation.

Shallit \cite{Sha04} has built an infinite binary word avoiding $7/3^+$-powers
 and all squares of period at least $7$.
His word contains $18$ squares.

Our infinite binary word avoids the same powers but contains only $12$ squares,
 the largest having period $8$.
As before the proof relies on a pair of morphisms satisfying suitable properties.
Both morphisms are almost uniform (up to one unit).
The first morphism is weakly square-free on a $6$-letter alphabet, and the
 second does not even correspond to a uniquely-decipherable code but admits
 a unique decoding on the words produced by the first.
To get the morphisms, we first examined carefully the structure of long words
 satisfying the conditions and obtained by backtracking computation.
Then, we inferred the morphisms from the regularities found in the words.

After introducing the definitions and main results in the next section, we provide a
 weakly square-free morphism and the infinite square-free word on $6$ letters it generates
 in Section~\ref{sect-wsf}.
Section~\ref{sect-trans} shows how this word is translated into an infinite
 binary word satisfying the constraints.
In the conclusion we define the new notion of finite-repetition threshold and
 state a conjecture on its value for the 3-letter alphabet.

\section{Repetitions in binary words}

A word is a sequence of letters drawn from a finite alphabet.
We consider the binary alphabet $B = \{\texttt{0},\texttt{1}\}$,
 the ternary alphabet $A_3 = \{\texttt{a},\texttt{b},\texttt{c}\}$, and the $6$-letter
 alphabet $A_6 = \{\texttt{a},\texttt{b},\texttt{c},\texttt{d},\texttt{e},\texttt{f}\}$.

A square is a word of the form $uu$ where $u$ is a nonempty (finite) word.
A word has period $p$ if its letters at distance $p$ are equal.
The exponent of a nonempty word is the quotient of its length
 over its smallest period.
Thus, a square is any word with an even integer exponent.

In this article we consider infinite binary words
 in which a small number of squares occur.

The maximal length of a binary word containing less than three square is finite.
It can be checked that it is $18$, e.g. $\texttt{010011000111001101}$ contains
 only $\texttt{00}$ and $\texttt{11}$.
But, as recalled above, this length is infinite if $3$ squares
 are allowed to appear in the word.
A simple proof of it relies on two morphisms $f$ and $h_0$ defined as follows.
The morphism $f$ is defined from $A_3^*$ to itself by
$$\cases{
  f(\texttt{a}) = \texttt{abc}, & \cr
  f(\texttt{b}) = \texttt{ac}, & \cr
  f(\texttt{c}) = \texttt{b}.
  }$$
It is known that the infinite word $\mathbf{f} = f(\texttt{a})^\infty$ it generates
 is square-free (see \cite[Chapter 2]{Lot97}).
The morphism $h_0$ is from $A_3^*$ to $B^*$ and defined by
$$\cases{
  h(\texttt{a}) = \texttt{01001110001101}, & \cr
  h(\texttt{b}) = \texttt{0011}, & \cr
  h(\texttt{c}) = \texttt{000111}.
  }$$
Then the result is a consequence of the next statement.

\begin{theorem}[\cite{BC10-3sq}]\label{theo-1}
The infinite word $\mathbf{h_0} = h_0(f(\mathtt{a})^\infty)$ contains the $3$ squares
 $\texttt{00}$, $\texttt{11}$ and $\texttt{1010}$ only.
 The cubes $\texttt{000}$ and $\texttt{111}$ are the only factors occurring in
 $\mathbf{h}$ and of exponent larger than $2$.
\end{theorem}

It is impossible to avoid $2^+$-powers and keep a bounded number of squares.
As proved by Karhum{\"a}ki and Shallit \cite{KarhumakiS04}, the exponent has to
 go up to $7/3$ to allow the property.

In the two following sections we define two morphisms and derive the properties that we
 need to prove the next statement.

\begin{theorem}\label{theo-2}
There exists an infinite binary word whose factors have an exponent at most $7/3$
 and that contains $12$ squares, the fewest possible.
\end{theorem}

Our infinite binary word contain the $12$ squares
 $\texttt{0}^2$, $\texttt{1}^2$, $(\texttt{01})^2$, $(\texttt{10})^2$,
 $(\texttt{001})^2$, $(\texttt{010})^2$, $(\texttt{011})^2$,
 $(\texttt{100})^2$, $(\texttt{101})^2$, $(\texttt{110})^2$,
 $(\texttt{01101001})^2$, $(\texttt{10010110})^2$,
 and the two words $\texttt{0110110}$ and $\texttt{1001001}$ of exponent $7/3$.

Proving that it is impossible to have less than $12$ squares in the previous
 statement results from the next table.
It has been obtained by pruned backtracking sequential computation that avoids
 exhaustive search.
It shows the maximal length of binary words whose factors have an exponent at
 most $7/3$, for each number $s$ of squares, $0 \leq s \leq 11$.
\[\begin{array}{|lcccccccccccc|}
      \hline
      s&=0&1&2&3&4&5&6&7&8&9&10&11 \\
      \hline
      \ell(s)&=3&5& 8 & 12 & 14 & 18 & 24 & 30 & 37 & 43 & 83 & 116  \\
      \hline
      \end{array}\]

\section{A weakly square-free morphism on six letters}
\label{sect-wsf}

In this section we consider a specific morphism used for the proof of Theorem~\ref{theo-2}.
It is called $g$ and defined from $A_6^*$
 to itself by:
$$\cases{
  g(\texttt{a}) = \texttt{abac}, & \cr
  g(\texttt{b}) = \texttt{babd}, & \cr
  g(\texttt{c}) = \texttt{eabdf}, &\cr
  g(\texttt{d}) = \texttt{fbace}, &\cr
  g(\texttt{e}) = \texttt{bace}, &\cr
  g(\texttt{f}) = \texttt{abdf}.
}$$

We prove below that the morphism is weakly square-free in the sense that
 $\mathbf{g} = g^\infty(\texttt{a})$ is an infinite square-free word, that is,
 all its finite factors have an exponent smaller than $2$.
Note that however it is not square-free since for example $g(\texttt{cf}) = \texttt{eabdfabdf}$
 contains the square $(\texttt{abdf})^2$.
This prevents from using characterisation of square-freeness of the morphism,
 or equivalently of the fixed points of the morphism.
As far as we know only an ad hoc proof is possible.

The set of codewords $g(a)$'s ($a \in A_6$) is a prefix code and therefore a
 uniquely-decipherable code.
Note also that any occurrence of $\texttt{abac}$ in $g(w)$, for $w \in A_6^*$,
 uniquely corresponds to an occurrence of $\texttt{a}$ in $w$.
The proof below relies on the fact that not all doublets and triplets
 (words of length $2$ and $3$ respectively)
 occur in $\mathbf{g}$, as the next statements show.

\begin{lemma}\label{lemm-1}
The set of doublets occurring in $\mathbf{g}$ is
 $$D = \{\mathtt{ab}, \mathtt{ac}, \mathtt{ba}, \mathtt{bd}, \mathtt{cb}, \mathtt{ce},
 \mathtt{da}, \mathtt{df}, \mathtt{ea}, \mathtt{fb}\}.$$
\end{lemma}
\begin{proof}
Note that all letters of $A_6$ appear in $\mathbf{g}$.
Then doublets
 \texttt{ab}, \texttt{ac}, \texttt{ba}, \texttt{bd},
 \texttt{ce}, \texttt{df}, \texttt{ea}, \texttt{fb}
 appear in $\mathbf{g}$ because they appear in the images of one letter.
The images of these doublets generate two more doublets, \texttt{cb} and \texttt{da},
 whose images do not create new doublets.
\end{proof}

\begin{lemma}\label{lemm-2}
\item The set of triplets in $\mathbf{g}$ is
 $$T = \{\mathtt{aba}, \mathtt{abd}, \mathtt{acb}, \mathtt{ace},
 \mathtt{bab}, \mathtt{bac}, \mathtt{bda}, \mathtt{bdf}, \mathtt{cba},
 \mathtt{cea}, \mathtt{dab}, \mathtt{dfb}, \mathtt{eab}, \mathtt{fba}\}.$$
\end{lemma}
\begin{proof}
Triplets appear in the images of a letter or of a doublet.
Triplets found in images of one letter are:
 \texttt{aba}, \texttt{abd}, \texttt{ace}, \texttt{bab},
 \texttt{bac}, \texttt{bdf}, \texttt{eab}, \texttt{fba}.
The images of doublets occurring in $\mathbf{g}$, in set $D$ of Lemma~\ref{lemm-1},
 contain the extra triplets:
 \texttt{acb}, \texttt{bda}, \texttt{cba}, \texttt{cea}, \texttt{dab}, \texttt{dfb}.
\end{proof}

To prove that the infinite word $\mathbf{g}$ is square-free we first show that it
 contains no square with less than four occurrences of the word
 $g(\texttt{a})=\texttt{abac}$.
Then, we show it contains no square with at least four occurrences of it.
The word \texttt{abac} is chosen because its occurrences in $\mathbf{g}$
 correspond to $g(\texttt{a})$ only, so they are used to synchronise the parsing of
 the word according to the codewords $g(\texttt{a})$'s.

\begin{lemma}\label{lemm-3}
No square in $\mathbf{g}$ can contain less than four occurrences of $\mathtt{abac}$.
\end{lemma}
\begin{proof}
Assume by contradiction that a square $ww$ in $\mathbf{g}$ contains
 less than four occurrences of \texttt{abac}.
Let $x$ be the shortest word whose image by $g$ contains $ww$.

Then $x$ is a factor of $\mathbf{g}$ that belongs to
 the set $\texttt{a}((A_6\setminus\{\texttt{a}\})^*\texttt{a})^5$.
Since two consecutive occurrences of \texttt{a} in $\mathbf{g}$ are
 separated by a string of length at most $4$ (the largest such string is indeed
 \texttt{bdfb} as a consequence of Lemma~\ref{lemm-2}), the set is finite.

The square-freeness of all these factors has been checked via an elementary implementation
 of the test, which proves the result.
\end{proof}

\begin{proposition}\label{prop-1}
No square in $\mathbf{g}$ can contain at least four occurrences of $\mathtt{abac}$.
\end{proposition}

\begin{table}
\caption{Gaps of \texttt{abac}: words between consecutive occurrences of \texttt{abac}
 in $\mathbf{g}$.
They are images of gaps between consecutive occurrences of \texttt{a}.
}\label{tabl-gap1}
$$\begin{array}{|lclr|}
\hline
 g(\texttt{b})  &=& \texttt{babd}     & 4  \\
 g(\texttt{cb}) &=& \texttt{eabdfbabd} & 9 \\
 g(\texttt{bd})  &=& \texttt{babdfbace}     & 9 \\
 g(\texttt{ce}) &=& \texttt{eabdfbace} &9 \\
 g(\texttt{bdfb}) &=& \texttt{babdfbaceabdfbabd} & 17 \\
\hline
\end{array}$$
\end{table}

\begin{proof}
The proof is by contradiction:
 let $k$ be the maximal integer for which $g^k(\texttt{a})$ is square-free
 and let $ww$ be a square occurring in $g^{k+1}(\texttt{a})$ and containing at least
 $4$ occurrences of \texttt{abac}.
Distinguishing several cases according to the words between consecutive occurrences
 of \texttt{abac} (see Table \ref{tabl-gap1}),
 we deduce that $g^k(\texttt{a})$ is not square-free, the contradiction.

The square $ww$ can be written
$$\underbrace{v_0(\texttt{abac} \cdots \texttt{abac})u_1}%
\underbrace{v_1(\texttt{abac} \cdots \texttt{abac})u_2}$$
 where $v_0$, $u_1$, $v_1$, and $u_2$ contain no occurrence of \texttt{abac}.
It occurs in the image of a factor of $\mathbf{g}$.
The central part of $w$ starting and ending with \texttt{abac} is the image
 of a unique word $U$ factor of $g^{k}(\texttt{a})$ due to the code property:
$$g(U) = v_0^{-1}wu_1^{-1} = v_1^{-1}wu_2^{-1}.$$
We split the proof in two parts according to whether \texttt{abac} occurs in
 $u_1v_1$ or not.

\paragraph{No \texttt{abac} in $u_1v_1$.}
We consider five cases according to the value of $u_1v_1$, the gap of \texttt{abac}
 (see Table~\ref{tabl-gap1}).

\begin{enumerate}
\item $u_1v_1 = \texttt{babd}$ corresponds to $g(\texttt{b})$ only.
If either $u_1$ or $v_1$ is empty, then $v_0$ or $u_2$ is $g(\texttt{b})$, in
 either case we get $\texttt{b}U\texttt{b}U$ or $U\texttt{b}U\texttt{b}$
 that are squares.
Else $v_0$ has a suffix \texttt{d} so it belongs to $g(\texttt{b})$,
 and again $\texttt{b}U\texttt{b}U$ is a square in $\mathbf{g}$.

\item $u_1v_1 =\texttt{eabdfbabd}$ corresponds to $g(\texttt{cb})$ only.
An occurrence of \texttt{cb} always belongs to $g(\texttt{ab})$ therefore
 $U$ has a prefix \texttt{abd} and a suffix \texttt{aba}, and
 the letter after \texttt{aba} is \texttt{c}.
If $v_1$ is empty, $u_2$ has a prefix
 \texttt{eabdfbabd} so it is $g(\texttt{cb})$ and again $U\texttt{cb}U\texttt{cb}$
 is a square.
If $v_1$ is not empty then $v_0$ has a suffix \texttt{d}, suffix of $g(\texttt{b})$,
 therefore $\texttt{b}U\texttt{cb}U\texttt{c}$ is a square.

\item $u_1v_1 =  \texttt{babdfbace}$ corresponds to $g(\texttt{bd})$.
The word \texttt{abda} is a factor of $g(\texttt{ba})$ only so $U$ has a prefix
 \texttt{aba} and a suffix \texttt{ba}.
If $|u_1|=0$, $v_0 = \texttt{babdfbace}$ can only be $g(\texttt{bd})$
 so $\texttt{bd}U\texttt{bd}U$ is a square.
Otherwise $u_2$ must have a prefix \texttt{b}; since $U$ has a suffix
 \texttt{ba} the next letter after it is either \texttt{b} or
 \texttt{c}; as only $g(\texttt{b})$ is prefixed by \texttt{b} the letter
 is \texttt{b} so $u_2$ has a prefix or is a prefix of $g(\texttt{b})$,
 and we know that \texttt{bab} is always followed by \texttt{d} thus
 $U\texttt{bd}U\texttt{bd}$ is a square.

\item $u_1v_1 =\texttt{eabdfbace}$ corresponds to $g(\texttt{ce})$ only.
If $u_1$ is empty, $v_0$ is $g(\texttt{ce})$ so $\texttt{ce}U\texttt{ce}U$
 is a square.
Otherwise, $u_2$ has a prefix or is a prefix of $g(\texttt{c})$;
 the next letter after $g(\texttt{c})$ is either \texttt{b} or \texttt{e};
 (see Lemma~\ref{lemm-1});
 if it is \texttt{b} the right-most $U$ has a suffix \texttt{aba}
 but the left-most $U$ has a suffix \texttt{fba}, which cannot be.
Therefore the letter after \texttt{c} is \texttt{e} and $U\texttt{ce}U\texttt{ce}$
 is a square.

\item $u_1v_1 =  \texttt{babdfbaceabdfbabd}$.
If $|v_1|>12$, $v_0$ has a suffix $g(\texttt{dfb})$ and the letter before it
 is \texttt{b}, so $\texttt{bdfb}U\texttt{bdfb}U$ is a square.
If $0 < |v_1| \leq 12$, then $|u_1| \geq 5$, so $u_2$ has a prefix or is a prefix
 of $g(\texttt{bd})$ so the next letter is either \texttt{a} or \texttt{f}.
If it is \texttt{a} the right-most $U$ has a suffix \texttt{ba} but $v_0$
 is a suffix of or has a suffix $g(\texttt{b})$; the letter before it is either
 $g(\texttt{c})$ or $g(\texttt{f})$; if it is \texttt{c} then $U$ has a prefix
 \texttt{abd} and \texttt{bdfbabd} is from the concatenation of
 $g(\texttt{c})$ and $g(\texttt{b})$ or $g(\texttt{dfb})$;
 in either case the left occurrence of $U$ will have
 \texttt{ea} as a suffix, a contradiction since $\texttt{fb}U\texttt{bdfb}U\texttt{bd}$ and
 $U\texttt{bdfb}U\texttt{bdfb}$ are both squares.
\end{enumerate}

\paragraph{An occurrence of \texttt{abac} in $u_1v_1$.}
Then the suffix of $u_1$ is either \texttt{aba}, \texttt{ab} or \texttt{a} while the
 respective prefix of $v_1$ is \texttt{c}, \texttt{ac} or \texttt{bac}.

Note that \texttt{c} is followed either by \texttt{b} or \texttt{e} (Lemma~\ref{lemm-1})
 and that \texttt{cb} occurs only in the image of \texttt{ab}.
Then if the occurrence of \texttt{abac} is followed by \texttt{b},
 the occurrence of \texttt{cb} in $v_0$ is preceded by \texttt{aba}, and then
 there is a square starting 1, 2 or 3 positions before the occurrence of $ww$,
 which brings us back to the first case.
Therefore, \texttt{abac} is followed by \texttt{e}.

The occurrence of \texttt{abace} comes from $g(\mathtt{ac})$, and by Lemma~\ref{lemm-2}
 $u_1v_1$ contains an occurrence of $g(\mathtt{bac})$.
So, the occurrence of \texttt{abace} is preceded by \texttt{d}, and since \texttt{da}
 occurs only in the image of \texttt{ba}, the occurrence of \texttt{da} in $u_2$
 is followed by \texttt{bac}, which yields a square starting 1, 2 or 3 positions
 after the occurrence of $ww$.
Again this takes us back to the first case.

In all cases we deduce the existence of a square in $g^k(\texttt{a})$, which is
 a contradiction with the definition of $k$.
Therefore there is no square in $\mathbf{g}$ containing at least
 four occurrences of \texttt{abac}.
\end{proof}

The next corollary is a direct consequence of Lemma~\ref{lemm-3}
 and Proposition~\ref{prop-1}.

\begin{corollary}\label{coro-1}
The infinite word $\mathbf{g}$ is square-free, or equivalently, the morphism $g$
 is weakly square-free.
\end{corollary}

\section{Binary translation}
\label{sect-trans}

The second part of the proof of Theorem~\ref{theo-2} consists in showing that the
 special infinite square-free word on $6$ letters introduced in the previous section
 can be transformed into the desired binary word.
This is done with a second morphism $h$ from $A_6^*$ to $B^*$ defined by
$$\cases{
  h(\mathtt{a}) = \mathtt{10011}, & \cr
  h(\mathtt{b}) = \mathtt{01100}, & \cr
  h(\mathtt{c}) = \mathtt{01001}, &\cr
  h(\mathtt{d}) = \mathtt{10110}, &\cr
  h(\mathtt{e}) = \mathtt{0110}, &\cr
  h(\mathtt{f}) = \mathtt{1001}.
  }$$
Note that the codewords of $h$ do not form a prefix code, nor a suffix code,
 nor even a uniquely-decipherable code!
We have for example $g(\mathtt{ae}) = \mathtt{10011}\cdot\mathtt{0110}
 = \mathtt{1001}\cdot\mathtt{10110} = g(\mathtt{fd})$.
However, parsing the word $h(y)$ when $y$ is a factor of $\mathbf{g}$ is unique
 due to the absence of some doublets and triplets in it
 (see Lemmas~\ref{lemm-1} and \ref{lemm-2}).
For example $\mathtt{fd}$ does not occur, which induces the unique parsing of
 $\mathtt{100110110}$
 as $\mathtt{10011}\cdot\mathtt{0110}$.

\begin{proposition}\label{prop-2}
The infinite word $\mathbf{h} = h(g^\infty(\mathtt{a}))$ contains no factor of
 exponent larger than $7/3$.
It contains the $12$ squares
 $\mathtt{0}^2$, $\mathtt{1}^2$, $(\mathtt{01})^2$,
 $(\mathtt{10})^2$, $(\mathtt{001})^2$, $(\mathtt{010})^2$,
 $(\mathtt{011})^2$, $(\mathtt{100})^2$, $(\mathtt{101})^2$,
 $(\mathtt{110})^2$, $(\mathtt{01101001})^2$, $(\mathtt{10010110})^2$
 only.
Words $\mathtt{0110110}$ and $\mathtt{1001001}$ are the only factors
 with an exponent larger than $2$.
\end{proposition}

\begin{figure}
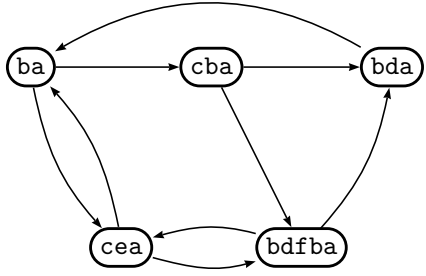

\begin{center}
\VCDraw[2]{
\begin{VCPicture}{(0,1)(6,4)}
\StateVar[\mathtt{ba}]{(1,3)}{1}
\StateVar[\mathtt{cba}]{(3,3)}{2}
\StateVar[\mathtt{bda}]{(5,3)}{3}
\StateVar[\mathtt{cea}]{(2,1)}{4}
\StateVar[\mathtt{bdfba}]{(4,1)}{5}
\Edge{1}{2}{}
\Edge{2}{3}{}
\LArcR{3}{1}{}
\ArcR{1}{4}{}
\ArcR{4}{1}{}
\ArcR{4}{5}{}
\ArcR{5}{4}{}
\Edge{2}{5}{}
\ArcR{5}{3}{}
\end{VCPicture}
}
\end{center}
\caption{Graph showing immediate successors of gaps in the word $\mathbf{g}$:
 a suffix of it following an occurrence of $\mathtt{a}$ is the label of an infinite path.
}\label{figu-graph}
\end{figure}

The proof is based on the fact that occurrences of $\mathtt{10011}$ in $\mathbf{h}$
 identify occurrences of $\mathtt{a}$ in $\mathbf{g}$ and on the unique parsing
 mentioned above.
It proceeds by considering several cases according to the gaps between consecutive
 occurrences of $\mathtt{10011}$ (see Table~\ref{tabl-gap2}), associated with gaps
 between consecutive occurrences of $\mathtt{a}$ in $\mathbf{g}$,
 which leads to analyse paths in the graph of Figure~\ref{figu-graph}.

\begin{table}
\caption{Gaps between consecutive occurrences of $\mathtt{10011}$ in $\mathbf{h}$.
}\label{tabl-gap2}
$$\begin{array}{|lclr|}
\hline
 h(\mathtt{b})    &=& \mathtt{01100}               &  5  \\
 h(\mathtt{cb})   &=& \mathtt{0100101100}          & 10 \\
 h(\mathtt{bd})   &=& \mathtt{0110010110}          & 10 \\
 h(\mathtt{ce})   &=& \mathtt{010010110}           &  9 \\
 h(\mathtt{bdfb}) &=& \mathtt{0110010110100101100} & 19 \\
\hline
\end{array}$$
\end{table}

\begin{proof}
We show that if $\mathbf{h}$ would contain a square not in the list it would come
 from a square in $\mathbf{g}$,
 which cannot be since $\mathbf{g}$ is square-free (Corollary~\ref{coro-1}).

Suppose $\mathbf{h}$ contains the square $w^2$.
It is a factor of $h(g^k(\mathtt{a}))$, for some integer $k$ and can be written
 $\underbrace{v_0(h(\mathtt{a})\cdots h(\mathtt{a}))u_1}
 \underbrace{v_1(h(\mathtt{a})\cdots h(\mathtt{a}))u_2}$.
The central part of $w$ is the image of a unique square-free factor $U$ of
 $g^k(\mathtt{a})$ due to the unique parsing mentioned above:
$$h(U) = (h(\mathtt{a})\cdots h(\mathtt{a})) = v_0^{-1}wu_1^{-1} = v_1^{-1}wu_2^{-1}.$$

We proceed through different cases as in the proof of Proposition~\ref{prop-1}.
\paragraph{No $h(\mathtt{a})$ in $u_1v_1$.}
\begin{enumerate}
\item $u_1v_1 = \mathtt{01100}$ corresponds to $h(\mathtt{b})$ only.

If  $|v_1|>1$, then $v_0$ belongs to $h(\mathtt{b})$, $\mathtt{b}U\mathtt{b}U$ is a square.
Else $|u_1|\geq 4$ so $u_2$ belongs to $h(\mathtt{b})$,
 it cannot belong to $h(\mathtt{e})$ since $\mathtt{ae}$ is not a factor of
 $\mathbf{g}$, therefore $U\mathtt{b}U\mathtt{b}$ is a square of $\mathbf{g}$.

\item $u_1v_1 = \mathtt{0110010110}$ corresponds to $h(\mathtt{bd})$.

$$v_0\underbrace{(h(\mathtt{a})\cdots h(\mathtt{a}))}h(\mathtt{bd})\underbrace{(h(\mathtt{a})\cdots h(\mathtt{a}))}u_2$$
 the word $\mathtt{abda}$ is a factor of $g(\mathtt{ba})$ only,
 so $U$ has a prefix $\mathtt{abac}$ and a suffix $\mathtt{ba}$
 (Note that $U$ cannot be $\mathtt{aba}$ since $\mathtt{ababdaba}$ is not a factor of
 $\mathbf{g}$).
$$v_0\underbrace{(h(\mathtt{abac})\cdots h(\mathtt{ba}))}h(\mathtt{bd})%
\underbrace{(h(\mathtt{abac})\cdots h(\mathtt{ba}))}u_2$$
If $u_2$ comes from or has a prefix  $h(\mathtt{b})$ then the letter after $\mathtt{bab}$
 is always $\mathtt{d}$ so we have the square $U\mathtt{bd}U\mathtt{bd}$.
Then $u_2$ is a prefix of or has a prefix $h(\mathtt{c})$, the longest common prefix (LCP)
 of $h(\mathtt{c})$ and $h(\mathtt{b})$ is $\mathtt{01}$,
 so $v_0$ has a suffix $\mathtt{10010110}$, which
 is a suffix of $h(\mathtt{bd})$ or $h(\mathtt{ce})$.
If $v_0$ comes from  $h(\mathtt{bd})$ then we have the square $\mathtt{bd}U\mathtt{bd}U$.
So  $v_0$ is a suffix of  $h(\mathtt{ce})$
$$h(\mathtt{ce})\underbrace{(h(\mathtt{abac})\cdots h(\mathtt{ba}))}h(\mathtt{bd})%
\underbrace{(h(\mathtt{abac})\cdots h(\mathtt{ba}) )}h(\mathtt{c}).
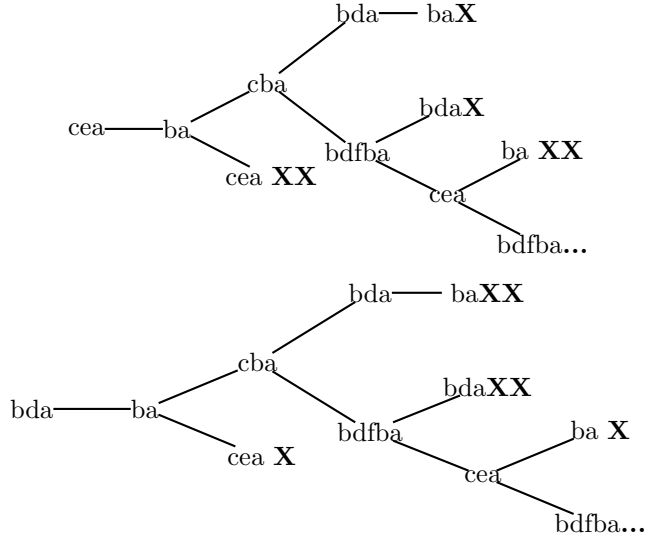
$$

\begin{center}
\pstree[treemode=R,treesep=1cm,levelsep=1.2cm]
{\Tr{cea}}{%
 \pstree{\Tr {ba}}{%

    \pstree{\Tr{cba}}{%
     \pstree{\Tr{bda}}{%
     \Tr*{ ba\bf X}}
	    \pstree{\Tr{bdfba}}{%
	    \Tr*{ bda\bf X}
	     \pstree{\Tr{cea}}{%
	       \Tr*{ ba \bf XX}
         \Tr*{ bdfba\bf ...}
 		     }}}

		     \Tr*{ cea \bf XX}		
		}}
\end{center}

\begin{center}
\pstree[treemode=R,treesep=1cm,levelsep=1.5cm]
{\Tr{bda}}{%
 \pstree{\Tr {ba}}{%

    \pstree{\Tr{cba}}{%
     \pstree{\Tr{bda}}{%
     \Tr*{ ba\bf XX}}
	    \pstree{\Tr{bdfba}}{%
	    \Tr*{ bda\bf XX}
	     \pstree{\Tr{cea}}{%
	       \Tr*{ ba \bf X}
         \Tr*{ bdfba\bf ...}
 		     }}}

		     \Tr*{ cea \bf X}		
		}}
\end{center}

The sign XX shows that the particular branch of the trie terminates because
 either a square occurs or the sequence is not a factor of $\mathbf{g}$.
The sign X on the other hand represents the termination of a particular branch
 as a consequence of the discontinuation of the corresponding branch in the other trie.
If we continue these tries we will have:
$$\mathtt{ce}\underbrace{\tt\,abac\,babd\,fbace\,abdf\,babd\,abac\,eabdf\,bace\,abac\,babd\,abac\,eabdf\dots ba}$$
$$\mathtt{bd}\underbrace{\tt\,abac\,babd\,fbace\,abdf\,babd\,abac\,eabdf\,bace\,abac\,babd\,abac\,eabdf\dots ba}\mathtt{c}$$
 which is the image of
$$\mathtt{eabdf\,bace}\underbrace{\tt \,abac\dots abac} \mathtt{babd\,fbace}%
\underbrace{\tt \,abac\dots abac}\mathtt{\,e}$$
 itself image of
$$\mathtt{ce}\underbrace{\tt a\dots a}\mathtt{bd}\underbrace{\tt a\dots a}\mathtt{c}$$
 so we have the same situation as at the starting point; but $U$ is shorter in this case,
 therefore if we continue this process we should have
$$\mathtt{ce\,abac\,babd\,fbace\,abdf\,babd\,\,abac\,babd\,fbace\,abdf\,bace\,a}$$
but $\mathtt{abdf\,bace}$ is the image of $\mathtt{fe}$ that is not in $D$
 (Lemma~\ref{lemm-1}).

\item $u_1v_1 = \mathtt{0100101100}$ corresponds to $h(\mathtt{cb})$.

The word $\mathtt{acba}$ is a factor of $g(\mathtt{ab})$ only,
 so $U$ has a prefix $\mathtt{abd}$ and a suffix $\mathtt{aba}$:
$$v_0\underbrace{(h(\mathtt{abd})\dots h(\mathtt{aba}) ) }h(\mathtt{cb})%
\underbrace{(h(\mathtt{abd})\dots h(\mathtt{aba}) )}u_2$$
The word $u_2$ comes from or has a prefix  $h(\mathtt{c})$.
If the letter after it is $\mathtt{b}$, we have the square $U\mathtt{cb}U\mathtt{cb}$.

Otherwise $u_2$ comes from or has a prefix  $h(\mathtt{ce})$.
If $v_0$ comes from or has a suffix $h(\mathtt{b})$ then  we have the square
 $\mathtt{b}U\mathtt{cb}U\mathtt{c}$.

Therefore the letter before $U$ is $\mathtt{e}$ preceded by $\mathtt{c}$,
 i.e. the string before the left $U$ is $\mathtt{ce}$:
$$h(\mathtt{ce})\underbrace{(h(\mathtt{abd})\dots h(\mathtt{aba}) ) }h(\mathtt{cb})%
\underbrace{(h(\mathtt{abd})\dots h(\mathtt{aba}) )}h(\mathtt{ce}).$$

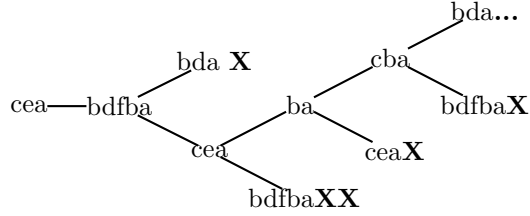
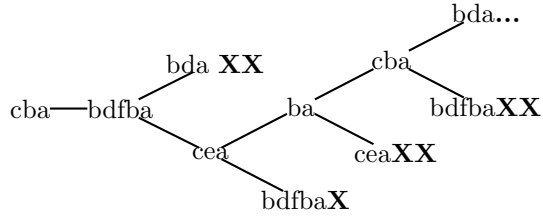
\begin{center}
\pstree[treemode=R,treesep=1cm,levelsep=1.2cm]
{\Tr{cea}}{%
 \pstree{\Tr {bdfba}}{%
  \Tr*{ bda \bf X}

    \pstree{\Tr{cea}}{%
     \pstree{\Tr{ba}}{%
     \pstree{\Tr{cba}}{
     \Tr*{ bda\bf ...}
     \Tr*{ bdfba\bf X}}
     \Tr*{ cea\bf X}}
         \Tr*{ bdfba\bf XX}
 		     }
		
		}}
\end{center}

\begin{center}
\pstree[treemode=R,treesep=1cm,levelsep=1.2cm]
{\Tr{cba}}{%
 \pstree{\Tr {bdfba}}{%
  \Tr*{ bda \bf XX}

    \pstree{\Tr{cea}}{%
     \pstree{\Tr{ba}}{%
     \pstree{\Tr{cba}}{
     \Tr*{ bda\bf ...}
     \Tr*{ bdfba\bf XX}}
     \Tr*{ cea\bf XX}}
         \Tr*{ bdfba\bf X}
 		     }
		
		}}
\end{center}

Now we have the same situation as in the previous case
$$h(g(\mathtt{ce}))\underbrace{(h(g(\mathtt{abac}))\dots h(g(\mathtt{ba})) ) }h(g(\mathtt{bd}))%
\underbrace{(h(g(\mathtt{abac}))\dots h(g(\mathtt{ba})) )}h(g(\mathtt{c})).$$

\item $u_1v_1 =\mathtt{010010110}$ corresponds to $h(\mathtt{ce})$ only.

Before $\mathtt{c}$ is always $\mathtt{ba}$ (Lemma~\ref{lemm-2})
 and after $\mathtt{e}$ is $\mathtt{ab}$ (Lemma~\ref{lemm-2}),
 so $\mathtt{ab}$ is a prefix of $U$ and $\mathtt{ba}$ is a suffix of $U$:
$$v_0\underbrace{(h(\mathtt{ab})\dots h(\mathtt{ba}) ) }h(\mathtt{ce})%
\underbrace{(h(\mathtt{ab})\dots h(\mathtt{ba}) )}u_2.$$

(i): $u_2$ belongs to $h(\mathtt{cb})$ since we cannot have $U\mathtt{ce}U\mathtt{ce}$
 and the letter after $\mathtt{c}$ is $\mathtt{b}$ or $\mathtt{e}$ (Lemma~\ref{lemm-1}):
$$v_0\underbrace{(h(\mathtt{ab})\dots h(\mathtt{ba}) ) }h(\mathtt{ce})%
\underbrace{(h(\mathtt{ab})\dots h(\mathtt{ba}) )}h(\mathtt{cb})$$
The letter before $\mathtt{bacb}$ is $\mathtt{a}$ so:
$$v_0\underbrace{(h(\mathtt{ab})\dots h(\mathtt{aba}))}h(\mathtt{ce})%
\underbrace{(h(\mathtt{ab})\dots h(\mathtt{aba}))}h(\mathtt{cb}).$$

\textbf{NOTE}: $U$ is not $\mathtt{aba}$ since $\mathtt{abaceabacb}$
 is not a factor of $g^k(\mathtt{a})$.

Now $\mathtt{abace}$ is a prefix of the image of $\mathtt{ac}$
 so $U$ has a prefix $\mathtt{abdf}$
 and the word before it is either $\mathtt{ce}$ or $\mathtt{b}$;
 the first choice gives the square $\mathtt{ce}U\mathtt{ce}U$ and the second choice:
$$h(\mathtt{b})\underbrace{(h(\mathtt{abdf})\dots h(\mathtt{aba}) ) }h(\mathtt{ce})%
\underbrace{(h(\mathtt{abdf})\dots h(\mathtt{aba}) )}h(\mathtt{cb}).$$

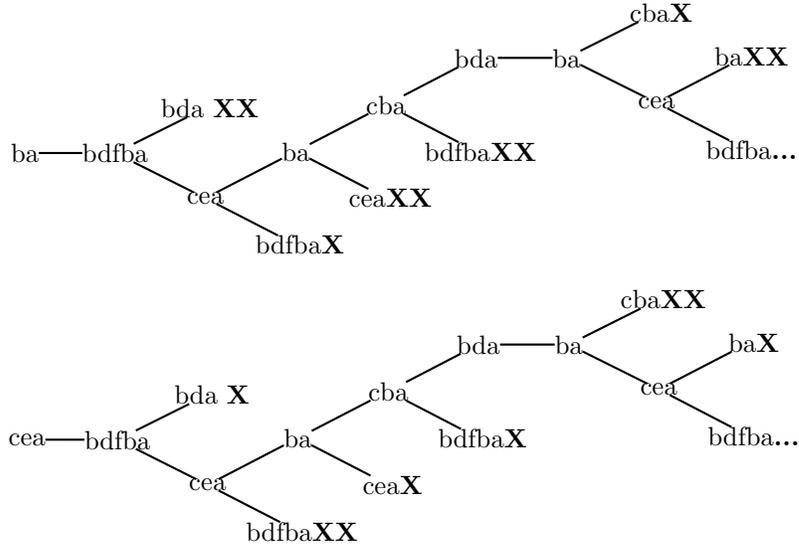
\begin{center}
\pstree[treemode=R,treesep=1cm,levelsep=1.2cm]
{\Tr{ba}}{%
 \pstree{\Tr {bdfba}}{%

  \Tr*{ bda \bf XX}

    \pstree{\Tr{cea}}{%
     \pstree{\Tr{ba}}{%
     \pstree{\Tr{cba}}{
     \pstree{\Tr{bda}}{%
     \pstree{\Tr{ba}}{%
     \Tr*{ cba\bf X}
     \pstree{\Tr{cea}}{%
          \Tr*{ ba\bf XX}
     \Tr*{ bdfba\bf ...}}}}
      \Tr*{ bdfba\bf XX}
 		     }
     \Tr*{ cea\bf XX}}
	
         \Tr*{ bdfba\bf X}
 		     }
		
		}}
\end{center}

\begin{center}
\pstree[treemode=R,treesep=1cm,levelsep=1.2cm]
{\Tr{cea}}{%
 \pstree{\Tr {bdfba}}{%

  \Tr*{ bda \bf X}

    \pstree{\Tr{cea}}{%
     \pstree{\Tr{ba}}{%
     \pstree{\Tr{cba}}{
     \pstree{\Tr{bda}}{%
     \pstree{\Tr{ba}}{%
     \Tr*{ cba\bf XX}
     \pstree{\Tr{cea}}{%
          \Tr*{ ba\bf X}
     \Tr*{ bdfba\bf ...}}}}
      \Tr*{ bdfba\bf X}
 		     }
     \Tr*{ cea\bf X}}
	
         \Tr*{ bdfba\bf XX}
 		     }
		
		}}
\end{center}

Now if we continue the above tries we get:
$$\mathtt{b}\underbrace{\tt abd\,fbace\,abac\,babd\,abac\,eabdf\,babd\,abac\,babd\,fbace\,abdf\,ba\dots ba}$$
$$\mathtt{ce}\underbrace{\tt abdf\,bace\,abac\,babd\,abac\,eabdf\,babd\,abac\,babd\,fbace\,abdf\,ba\dots ba}\mathtt{cb}$$
 which is the image of
$$\mathtt{bd}\underbrace{\tt \,abac\,babd\,fbace\,abdf\dots ba} \mathtt{ce}%
\underbrace{\tt \,abac\,babd\,fbace\,abdf\dots ba}\mathtt{b}.$$
This is the same situation as the next case and we will see that after going one step back
 it brings us back to this case again.
Now we are exactly in the same situation as at the beginning except that the length
 of the word $X=\mathtt{abdf}\dots\mathtt{a}$ is shorter than $U$.
Repeating this process enough times we should see that the word
$$\mathtt{ babd\,fbace\,abac\,babd\,abac\,eabdf\,bace\,abac\,babd\,aba}$$
 which is the image of $\mathtt{bdabaceaba}$, is not a factor of $g^k(\mathtt{a})$.

(ii): $u_2$ belongs to $h(\mathtt{b})$
 (the LCP of $h(\mathtt{c})$ and $h(\mathtt{b})$ is $\mathtt{01}$)
 so $v_0$ must have a suffix $\mathtt{0010110}$, which belongs to $h(\mathtt{bd})$
 because if it belongs to $h(\mathtt{ce})$ then $\mathtt{ce}U\mathtt{ce}U$ is a square.

$$h(\mathtt{bd})\underbrace{(h(\mathtt{ab})\dots h(\mathtt{ba}) ) }h(\mathtt{ce})%
\underbrace{(h(\mathtt{ab})\dots h(\mathtt{ba}) )}h(\mathtt{b}).
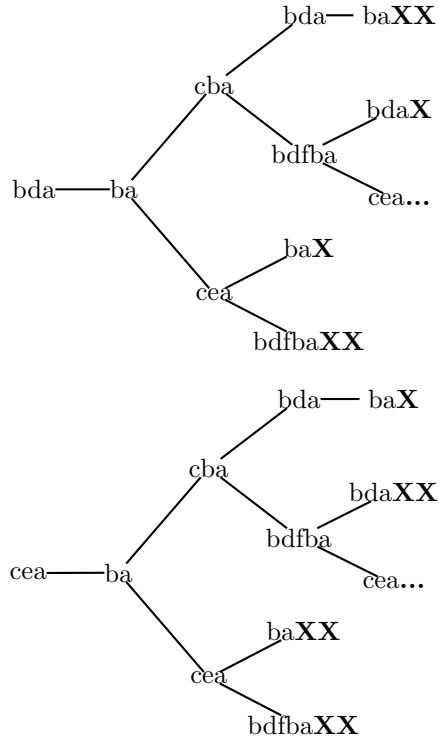
$$

\begin{center}
\pstree[treemode=R,treesep=1cm,levelsep=1.2cm]
{\Tr{bda}}{%
 \pstree{\Tr {ba}}{%
  \pstree{\Tr{cba}}{%
  \pstree{\Tr{bda}}{%
   \Tr*{ ba\bf XX}}
   \pstree{\Tr{bdfba}}{%
   \Tr*{ bda\bf X}
   \Tr*{ cea\bf ...}
    }}
    \pstree{\Tr{cea}}{%
     \Tr*{ ba\bf X}
	
      \Tr*{ bdfba\bf XX}
 		     }
		}}
\end{center}

\begin{center}
\pstree[treemode=R,treesep=1cm,levelsep=1.2cm]
{\Tr{cea}}{%
 \pstree{\Tr {ba}}{%
  \pstree{\Tr{cba}}{%
  \pstree{\Tr{bda}}{%
   \Tr*{ ba\bf X}}
   \pstree{\Tr{bdfba}}{%
   \Tr*{ bda\bf XX}
   \Tr*{ cea\bf ...}
    }}
    \pstree{\Tr{cea}}{%
     \Tr*{ ba\bf XX}
	
      \Tr*{ bdfba\bf XX}
 		     }	
		}}
\end{center}

Continuing this trie we have
$$\mathtt{bd}\underbrace{\tt abac\,babd\,fbace\,a\dots ba}\mathtt{ce}%
\underbrace{\tt \,abac\,babd\,fbace\,a\dots ba}\mathtt{bd}.$$
This is factor of
$g(\mathtt{b}\underbrace{\tt abdf\dots a}\mathtt{ce} \underbrace{\tt abdf\dots a}\mathtt{cb})$
 which is the previous case.

\item $u_1v_1 = \mathtt{0110010110100101100}$ corresponds to $h(\mathtt{bdfb})$ only.
This case is dealt with the same method.
$$u_0\underbrace{(h(\mathtt{a})\dots h(\mathtt{a}) ) }h(\mathtt{bdfb})%
\underbrace{(h(\mathtt{a})\dots h(\mathtt{a}) )}u_2.$$
If $u_2$ belongs to $h(\mathtt{c})$, the LCP of $h(\mathtt{c})$
 and $h(\mathtt{b})$ is $\mathtt{01}$ so $u_0$ must have a suffix $\mathtt{10010110100101100}$,
 therefore $u_0$ belongs to $h(\mathtt{bdfb})$.
But $\mathtt{bdfb}U\mathtt{bdfb}U$
 is a square and a factor of $g^k(a)$; a contradiction, so $u_2$ belongs to
 or has a prefix $h(\mathtt{b})$.
We have two choices here.

(i): the next word after the right occurrence of $U$ is $\mathtt{ba}$.
The LCP of $h(\mathtt{bd})$ and $h(\mathtt{ba})$ is $\mathtt{10}$,
 $u_0$ has suffix of $\mathtt{110100101100}$, so it either belongs to $h(\mathtt{dfb})$
 or $h(\mathtt{acb})$.
The first case gives that $\mathtt{dbf}U\mathtt{bdbf}U\mathtt{b}$
 is a square and a factor of $g^k(a)$, a contradiction.
So $u_0$ belongs to $h(\mathtt{acb})$:
$$h(\mathtt{acb})\underbrace{(h(\mathtt{abda})\dots h(\mathtt{a}) ) }h(\mathtt{bdf\,b})%
\underbrace{(h(\mathtt{abda})\dots h(\mathtt{a}) )}h(\mathtt{ba}).$$
Prefixes and suffixes of $U$ are determined only by looking at $D$ and $T$.
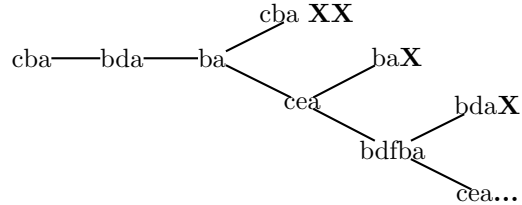
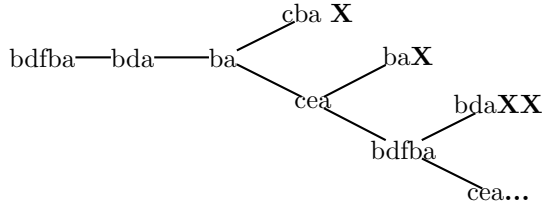
 \begin{center}
\pstree[treemode=R,treesep=1cm,levelsep=1.2cm]
{\Tr{cba}}{%
 \pstree{\Tr {bda}}{%

    \pstree{\Tr{ba}}{%
     \Tr*{ cba \bf XX}
     \pstree{\Tr{cea}}{%
     \Tr*{ ba\bf X}
	    \pstree{\Tr{bdfba}}{%
	    \Tr*{ bda\bf X}
	
         \Tr*{ cea\bf ...}
 		     }}}

		}}
\end{center}

\begin{center}
\pstree[treemode=R,treesep=1cm,levelsep=1.2cm]
{\Tr{bdfba}}{%
 \pstree{\Tr {bda}}{%

    \pstree{\Tr{ba}}{%
     \Tr*{ cba \bf X}
     \pstree{\Tr{cea}}{%
     \Tr*{ ba\bf X}
	    \pstree{\Tr{bdfba}}{%
	    \Tr*{ bda\bf XX}
	
         \Tr*{ cea\bf ...}
 		     }}}

		}}
\end{center}

We have:
$$\mathtt{abac\,babd}\underbrace{\tt\,abac\,eabdf\,bace\dots abac}\mathtt{babd\,fbace}$$
$$\mathtt{abdf\,babd}\underbrace{\tt \,abac\,eabdf\,bace\dots abac}\mathtt{babd\,fbace\,abac}$$
 which is the image of
$$\mathtt{ab}\underbrace{\tt ace\dots a}\mathtt{bdfb}%
\underbrace{\tt ace\dots a}\mathtt{bda}.$$
Now this is the next case so if we go back enough steps we should see that the length
 of $U$ decreases and at the end we get
$$\mathtt{ac\,babd\,abac\,eabdf\,babd\,abac\,eaba}$$
 but this is not a factor of $g^k(a)$, a contradiction.

(ii): the word after $U$ is $\mathtt{bd}$.
 Now here the only possible letter after $\mathtt{abd}$ is $\mathtt{a}$
 since if it is $\mathtt{f}$ it is a prefix of $\mathtt{fb}$ so we have
 $U\mathtt{bdfb}U\mathtt{bdfb}$, a contradiction.
As the LCP of $h(\mathtt{bdfb})$ and $h(\mathtt{bda})$ is $\mathtt{01100101101001}$
 $u_0$ must have a suffix $\mathtt{01100}$ so it can belong to $h(\mathtt{ab})$
 or $h(\mathtt{acb})$.

 (I):
$$h(\mathtt{ab})\underbrace{(h(\mathtt{a})\dots h(\mathtt{a}) ) }h(\mathtt{bdfb})%
\underbrace{(h(\mathtt{a})\dots h(\mathtt{a}) )}h(\mathtt{bda}).$$
 Only using $D$, $T$ and the Figure \ref{figu-graph} we can continue building $U$,
$$h(\mathtt{ab})\underbrace{(h(\mathtt{ace})\dots h(\mathtt{ba}) ) }h(\mathtt{bdfb})%
\underbrace{(h(\mathtt{acea})\dots h(\mathtt{ba}) )}h(\mathtt{bda}).$$
 Continuing further we get:
$$h(\mathtt{abac\,eabdf\,babd}\underbrace{\tt \,abac\,\dots abac}\mathtt{babd\,fbace\,abdf\,babd}%
\underbrace{\tt \,abac\dots abac}\mathtt{babda}).$$
This is the image of
$$h(g( \mathtt{acb}\underbrace{\tt a\dots a}\mathtt{bdfb}%
\underbrace{\tt a\dots a}\mathtt{ba}))$$
 and we are back to the case above.

 (II):
$$h(\mathtt{acb})\underbrace{(h(\mathtt{a})\dots h(\mathtt{a}) ) }h(\mathtt{bdfb})%
\underbrace{(h(\mathtt{a})\dots h(\mathtt{a}) )}h(\mathtt{bda}).$$
 Using the same method we build the word $U$:
$$\mathtt{ac\,b}\underbrace{\tt abd\dots ba}\mathtt{bd\,fb}%
\underbrace{\tt ace\,\dots ba}\mathtt{bd\,a}.$$
Here we cannot go further as $U$ cannot have $\mathtt{abd}$ nor $\mathtt{ace}$ as
 prefixes at the same time.
\end{enumerate}

\paragraph{An occurrence of $h(\mathtt{a})$ in $u_1v_1$.}
Looking at Figure \ref{figu-graph}, the image of the concatenation of two connected
 nodes (distance 1 arrow) are the possibilities for $u_1v_1h(\mathtt{a})$,
 but note that the second period of the square must start within $h(a)$,
 starting point of the arrow, otherwise it is one of the cases above.
If the lengths of both nodes are larger than $2$ then by unique parsing we are bound
 to have a square in $g^k(a)$ and get a contradiction.
So we have to consider only the four cases where one of the nodes is $\mathtt{ba}$:

\begin{enumerate}
\item $u_1v_1 = h(\mathtt{bacb}) = \mathtt{01100}\textbf{10011}\mathtt{0100101100}$,
 so $u_2$ must have a prefix $h(\mathtt{b})$ and $u_0$ a suffix of $h(\mathtt{cb})$,
 before $\mathtt{cb}$ is always $\mathtt{a}$, so $acbUbacbUb$ is a square in $g^k(a)$.

\item $u_1v_1 = h(\mathtt{bace}) = \mathtt{01100}\textbf{10011}\mathtt{010010110}$,
 so $u_2$ must have a prefix $h(\mathtt{b})$ and $u_0$ a suffix $h(\mathtt{ce})$,
 before $\mathtt{ce}$ is always $\mathtt{a}$, so $aceUbaceUb$ is a square in $g^k(a)$.

\item $u_1v_1 = h(\mathtt{ceab}) = \mathtt{010010110}\textbf{10011}\mathtt{01100}$,
 so $u_2$ must have a prefix of $h(\mathtt{ce})$ and $u_0$ a suffix of $h(\mathtt{b})$,
 after $\mathtt{ce}$ is always $\mathtt{a}$, so $bUceabUcea$ is a square in $g^k(a)$.

\item $u_1v_1 = h(\mathtt{bdab}) = \mathtt{0110010110}\textbf{10011}\mathtt{01100}$,
 so using tries as before shows that after enough backward iteration we should have
$$\mathtt{fbace\,abdf\,babd\,abac\,babd\,abac\,eabdf\,babd\,abac\,babd\,}$$
 which contains a square.
\end{enumerate}

In all cases the conclusion is that we get a square in $g^k(\mathtt{a})$, a contradiction
 with the definition of $k$.
This completes the proof of Proposition~\ref{prop-2}.
\end{proof}

Theorem~\ref{theo-2} follows immediately from Proposition~\ref{prop-2}.

\section{Conclusion}

The constraint on the number of squares imposed on binary words slightly differs from
 the constraint considered by Shallit \cite{Sha04}.
The squares occurring in his word have period smaller than $7$.
Our word contains less squares but their maximal period is $8$.

Looking at repetitions in words on larger alphabets, the subject introduces
 a new type of threshold, that we call the \textit{finite-repetitions threshold} (FRt).
For the alphabet of $a$ letters, $\mbox{FRt}(a)$ is defined as the smallest rational number
 for which there exists an infinite word avoiding $\mbox{FRt}(a)^+$-powers
 and containing a finite number of $r$-powers, where $r$ is Dejean's repetitive threshold.
Karhum{\"a}ki and Shallit results as well as ours show that $\mbox{FRt}(2) = 7/3$.
Our result additionally proves that the associated minimal number of squares is $12$.

Computation shows that the maximal length of $(7/4)^+$-free ternary word
 with only one $7/4$-repetition is $102$.
This leads us state the following conjecture, which has been tested up to length $20000$.

\begin{conjecture}
The \textit{finite-repetitions threshold} of 3-letter alphabet is $\frac{7}{4}$
 and the associated number of $\frac{7}{4}$-powers is $2$.\\
\end{conjecture}

Values for larger alphabets remain to be explored.

\small
\bibliographystyle{abbrv}
\bibliography{12squares}

\begin{thebibliography}{10}

\bibitem{BC10-3sq}
G.~Badkobeh and M.~Crochemore.
\newblock An infinite binary word containing only three distinct squares.
\newblock 2010.
\newblock Submitted.

\bibitem{CR11}
J.~D. Currie and N.~Rampersad.
\newblock A proof of {D}ejean's conjecture.
\newblock {\em Math. Comput.}, 80(274):1063--1070, 2011.

\bibitem{Dej72}
F.~Dejean.
\newblock Sur un th{\'e}or{\`e}me de {T}hue.
\newblock {\em J. Comb. Theory, Ser. A}, 13(1):90--99, 1972.

\bibitem{FS95}
A.~S. Fraenkel and J.~Simpson.
\newblock How many squares must a binary sequence contain?
\newblock {\em Electr. J. Comb.}, 2, 1995.

\bibitem{HN06}
T.~Harju and D.~Nowotka.
\newblock Binary words with few squares.
\newblock {\em Bulletin of the EATCS}, 89:164--166, 2006.

\bibitem{KarhumakiS04}
J.~Karhum{\"a}ki and J.~Shallit.
\newblock Polynomial versus exponential growth in repetition-free binary words.
\newblock {\em J. Comb. Theory, Ser. A}, 105(2):335--347, 2004.

\bibitem{Lot97}
M.~Lothaire, editor.
\newblock {\em Combinatorics on Words}.
\newblock Cambridge University Press, second edition, 1997.

\bibitem{RampersadSW05}
N.~Rampersad, J.~Shallit, and M.~{W}ei Wang.
\newblock Avoiding large squares in infinite binary words.
\newblock {\em Theor. Comput. Sci.}, 339(1):19--34, 2005.

\bibitem{Rao11}
M.~Rao.
\newblock Last cases of {D}ejean's conjecture.
\newblock {\em Theor. Comput. Sci.}, 412(27):3010--3018, 2011.

\bibitem{Sha04}
J.~Shallit.
\newblock Simultaneous avoidance of large squares and fractional powers in
  infinite binary words.
\newblock {\em Intl. J. Found. Comput. Sci}, 15(2):317--327, 2004.

\bibitem{Thue06}
A.~Thue.
\newblock {\"U}ber unendliche {Z}eichenreihen.
\newblock {\em Norske vid. Selsk. Skr. I. Mat. Nat. Kl. Christiana}, 7:1--22,
  1906.

\end{thebibliography}

\end{document}